\documentclass[elsarticle]{article}

\usepackage{multicol}
\usepackage{graphics} 		
\usepackage{epsfig} 		
\usepackage{amsmath} 		
\usepackage{amssymb}  
\usepackage{amsthm}

\newtheorem{theorem}{Theorem}[section]
\newtheorem{lemma}[theorem]{Lemma}
\newtheorem{proposition}[theorem]{Proposition}
\newtheorem{corollary}[theorem]{Corollary}

\title{\LARGE \bf
Decoherence Control and Purification of Two-dimensional Quantum Density Matrices under Lindblad Dissipation}

\author{Patrick Rooney, Anthony Bloch and Chitra Rangan
\thanks{P.R. and A.B. are at the Department of Mathematics, University of
Michigan, Ann Arbor, MI 48109, dprooney@umich.edu (P.R.) and
abloch@umich.edu (A.B.). Research is partly supported by NSF.
}
\thanks{C.R. is at the Department of Physics, University of Windsor, ON,
N9B 3P4. Canada, rangan@uwindsor.ca. Research is supported by
NSERC, Canada.
}
}

\begin{document}

\maketitle

\begin{abstract}
Control of quantum dissipative systems can be challenging because control variables are typically part of the system Hamiltonian, which can only generate motion along unitary orbits of the system. To transit between orbits, one must harness the dissipation super-operator. To separate the inter-orbit dynamics from the Hamiltonian dynamics for a generic two-dimensional system, we project the Lindblad master equation onto the set of spectra of the density matrix, and we interpret the location along the orbit to be a new control variable. The resulting differential equation allows us to analyze the controllability of a general two-dimensional Lindblad system, particularly systems where the dissipative term has an anti-symmetric part. We extend this to categorize the possible purifiable systems in two dimensions. 
\end{abstract}

\section{Introduction}

Recent decades have seen the application of mathematical control theory to quantum systems in both physics and chemistry, as technological advances have allowed for greater precision in manipulation of these systems \cite{huang}\cite{shapiro}\cite{tannor1}\cite{tannor2}\cite{ramakrishnaetal}. One particular area of interest is the possible construction of quantum computers, which have the power to perform algorithms not accessible to conventional computers. A major experimental obstacle to any implementation of such a computer, however, is the decoherence of the system under influence of the environment. While much progress has been made on the control of closed quantum systems \cite{dalessandro}\cite{albertini}\cite{rouchon}, work on open quantum systems has proved more challenging
 \cite{altafini2}\cite{altafini}\cite{sugny}\cite{schirmer}\cite{dirr}.  

One important issue is that controls are nearly always in the form of Hamiltonian operators. This type of control is unable to directly affect the purity of a state \cite{tannorbartana} or transfer the state between unitary orbits. To control purity, one must use the dissipative dynamics to move between orbits. To this end, we wish to derive a differential equation that captures only the inter-orbit dynamics, and collects the remaining dynamics (along the orbits) into a new control variable. This can be done if we assume arbitrary control over the Hamiltonian dynamics. The resulting differential equation can tell us how the location along the orbit affects the motion between orbits. In this paper, we show that this can be done for two-dimensional systems subject to Lindblad dissipation, and the formalism can be extended to consider the problem of purification (for related work in two dimensions, see \cite{altafini2}\cite{altafini}). The preliminaries of Lindblad dissipation are outlined in section 2, and the projection of the Lindblad differential equation onto the set of unitary orbits is discussed in section 3. In section 4, we analyze the controllability of this equation for various choices of system parameters $-$ in particular, the case where the anti-symmetric part of the dissipation is non-zero. In section 5, we present a theorem that specifies necessary and sufficient conditions for purifiability.

\section{Preliminaries}

The state of a closed quantum system is described by a norm-one vector in a complex Hilbert space that evolves according to the Schr\"{o}dinger equation:
\[ \frac{d}{dt} |\psi(t) \rangle = -iH|\psi(t)\rangle \]
\noindent In order for the norm to be preserved, the Hamiltonian operator $H$ must be Hermitian.
 An open quantum system, on the other hand, is described by a trace-one, positive-semidefinite operator $\rho$ on the Hilbert space, known as the density operator (or density matrix when working in finite dimensions, as we shall). The interpretation of this matrix is the following: an eigenvalue of $\rho$ is the probability that the system is in the corresponding eigenstate. Since the matrix is trace-one and positive-semidefinite, these eigenvalues are non-negative real numbers that sum to one. A state $|\psi\rangle$ in the closed system becomes a rank-one projection operator $|\psi\rangle\langle \psi|$\footnote{The bra-ket notation prescribes that a vector be written as $|a\rangle$ and its dual as $\langle a|$. Inner products are written $\langle a|b\rangle$ and outer products (or rank-one matrices) $|a\rangle\langle b|$.}. The Schr\"{o}dinger equation, when extended to the density matrix, becomes the von Neumann equation:
\[ \frac{d}{dt}\rho(t)  = [-iH, \rho(t)] \]

Certain relevant quantities are invariant under the von Neumann equation. The density matrix at any time can be written $\rho(t)=U(t)\rho(0)U^{-1}(t)$, where $U(t)$ is unitary. Since matrices at different times are similar, the eigenvalues are constant. The purity of the system, which is defined to be $\sqrt{tr(\rho^2)}$, is also invariant since it is the 2-norm of the vector of eigenvalues. This has implications for quantum control. Since control variables typically appear in the Hamiltonian only, the control dynamics cannot directly alter the probabilities, or purify the state (\emph{i.e.} achieve a purity of one).

However, a system that interacts with the environment will have non-Hamiltonian dynamics. In general, this will be an integro-differential equation, but if one assumes the dynamics depends only the present state and not its history (\emph{i.e.} the Markovian condition) and there is not explicit time dependence, the resulting differential equation is the Lindblad equation\cite{lindblad}\cite{breuer}:
\begin{equation} \frac{d}{dt}|\rho(t) \rangle = [-iH, \rho(t)] + \sum_{j=1}^{M} \left(L_j\rho L_j^\dagger -\frac{1}{2} \{L_j^\dagger L_j,\rho\}\right) \label{lindblad} \end{equation}
\noindent The Lindblad operators $\{L_j\}$ can be taken to be traceless, as adding a multiple of the identity $aI$ to $L_j$ is equivalent to adding an operator $\frac{i}{2}(\bar{a}L_j-aL_j^\dagger)$ to the Hamiltonian. An alternate equation, known as the Lindblad-Kossakowski equation, chooses a basis $\{l_j\}$ of the set of traceless $n$-dimensional matrices that is orthonormal relative to the inner product $(A,B)=tr(A^\dagger B)$:
\begin{equation} \frac{d}{dt}|\rho(t) \rangle = [-iH, \rho(t)] + \sum_{j,k=1}^{n^2-1} a_{jk}\left(l_j\rho l_k^\dagger -\frac{1}{2} \{l_k^\dagger l_j,\rho\}\right) \label{lindkoss}\end{equation}
\noindent where the coefficients $a_{jk}$'s form a positive-semidefinite matrix, known as the Gorini-Kossakowski-Sudarshan matrix.

A Lindblad operator can be thought as a stochastic jump with recoil. Under the influence of one Lindblad operator, a pure state $|\psi\rangle\langle\psi|$ in time $\delta t$ becomes a mixture of two states, $M_1|\psi\rangle\langle\psi|M_1^\dagger+M_2|\psi\rangle\langle\psi|M_2^\dagger$. Here, $M_1 = L \sqrt{\delta t}$ and $M_2=I-\frac{1}{2}L^\dagger L \delta t$. In other words, $|\psi\rangle$ jumps to the state $\frac{1}{\sqrt{\langle\psi|L^\dagger L|\psi\rangle}}L|\psi\rangle$ with probability $\sqrt{\langle\psi|L^\dagger L|\psi\rangle} \delta t$. This is a jump because the $\delta t$ appears in the probability only, meaning the destination state does not approach the original state as $\delta t\rightarrow 0^+$. Conversely, the second state in the mixture is $\frac{1}{\sqrt{1- \langle \psi|L^\dagger L|\psi \rangle\delta t}}(I-\frac{1}{2}L^\dagger L\delta t) |\psi\rangle$, which is infinitesimally close to the original state. In other words, depending on $L$, there may be an infinitesimal recoil needed to compensate for the jump process. When $L^\dagger L$ is a multiple of the identity (for example, when the Lindblad operator is a multiple of a Pauli matrix), the second state reduces to the original state $|\psi\rangle$, so that the jump is recoil-less.

\section{Projection of dynamics in two dimensions}

As mentioned in the introduction, control of open quantum systems typically involves control variables in the Hamiltonian. Hamiltonian operators, however, can only move states \emph{along} unitary orbits, and not \emph{between} orbits. The goal of this paper is to isolate the between-orbit dynamics for a generic two-dimensional system under Lindblad dissipation. Our starting point is the following control system: 

\begin{equation} \frac{d}{dt}\rho = \sum_{j=x,y,z}[-i u_j\sigma_j,\rho] + \frac{1}{2} \sum_{j,k=x,y,z} a_{jk}\left(\sigma_j\rho \sigma_k -\frac{1}{2} \{\sigma_k\sigma_j,\rho\}\right)  \label{bigDE} \end{equation}

\noindent where $\{ \sigma_j:j=x,y,z\}$ are the Pauli matrices. The controls $\{ u_j\}$ are unbounded and may take any value in $\mathbb{R}$. Note that we have chosen our set of control Hamiltonians to span $su(2)$. In other words, we can make any unitary operator up to a non-physical phase difference, and therefore we can move between any two states on a given unitary orbit in arbitrary time. We are neglecting any drift Hamiltonian $H_0=c_0I+\sum_{j=x,y,z} c_j\sigma_j$, since the component along the identity matrix does not contribute to the dynamics, and the components along the Pauli matrices can be treated by re-calibrating the control variables: $u_j\rightarrow u_j-c_j$.

The density operator can be written in terms of the Pauli matrices: $\rho=\frac{1}{2}(I+\sum_{j=x,y,z}n_j\sigma_j)$, where the $n_j$'s are components of the Bloch vector, such that \emph{i.e.} $n_x^2+n_y^2+n_z^2\le 1$. Substituting this expressions into the equation (\ref{bigDE}), we get:
\begin{align*} \frac{1}{2}\sum_j\frac{dn_j}{dt}\sigma_j &=  \sum_{j,k} [-iu_j\sigma_j , \frac{1}{2}n_k\sigma_k] + \frac{1}{4}\sum_{jk} a_{jk}[\sigma_j, \sigma_k] + \frac{1}{4}\sum_{jkl} a_{jk}n_l ( \sigma_j\sigma_l\sigma_k-\frac{1}{2}\{\sigma_k\sigma_j,\sigma_l\} ) \\
	&=  \sum_{j,k} (\frac{-i}{2}u_j n_k+\frac{1}{4}a_{jk}) [\sigma_j,\sigma_k]  +\frac{1}{4}\sum_{jkl} a_{jk}n_l (\sigma_j\sigma_l\sigma_k-\frac{1}{2}\{\sigma_k\sigma_j,\sigma_l\} ) 
\end{align*}
\noindent The Pauli matrices obey the relations 
\[ [\sigma_j,\sigma_k] = 2i\sum_l \epsilon_l \sigma_l \] 
\[ \{\sigma_j,\sigma_k\} = 2 \delta_{jk} I  \] 
\[ \sigma_j\sigma_l\sigma_k -\frac{1}{2}\{\sigma_k\sigma_j,\sigma_l\}= \delta_{kl}\sigma_j+\delta_{jl}\sigma_k - 2\delta_{jk}\sigma_l \]
\noindent Using these relations, the Lindblad-Kossakowski equation above becomes:
\begin{align*} \frac{1}{2}\sum_l\frac{dn_l}{dt}\sigma_l &=  \sum_{j,k,l} \epsilon_{jkl}u_j n_k\sigma_l+\sum_{j,k,l}\frac{1}{2}i a_{jk} \epsilon_{jkl}\sigma_l +  \frac{1}{4}\sum_{jl} (a_{jl}(n_j\sigma_l+n_l\sigma_j)-2a_{jj}n_l\sigma_l) 
\end{align*}
\noindent If we define $b_l=\sum ia_{jk}\epsilon_{jkl}$, and $a^S_{jk}=\frac{a_{jk}+a_{kj}}{2}$, we have
\begin{align*} \sum_l\frac{dn_l}{dt}\sigma_l &=  2\sum_{j,k,l} \epsilon_{jkl}u_j n_k\sigma_l+\sum_{l}b_l \sigma_l +  \sum_{jl} (a^S_{jl}n_j\sigma_l - a^S_{jj}n_l\sigma_l) 
\end{align*}
\noindent In vector notation, we can write:
\begin{equation} \frac{d\vec{n}}{dt} = \vec{b} + \vec{u} \times \vec{n} + (A^S-tr(A^S) I) \vec{n} \label{unproj} \end{equation}
\noindent where $A^S$ is the matrix with elements $a^S_{ij}$.

Now we want to decompose this equation into dynamics along and between unitary orbits.  $\rho$ has eigenvalues $\frac{1\pm r}{2}$, where $ r := |\vec{n}|$ and eigenvectors 
\[ |\psi_{\pm}\rangle := \sqrt{\frac{1+n_z}{2}}  |1\rangle + \frac{n_x+in_y}{\sqrt{2(1+n_z)}} |2\rangle \]
\noindent  Note the spectra correspond one-to-one with the values of $r$, the Bloch radius. It follows that the unitary orbits are concentric spheres, except for the completely mixed state, which corresponds to the point $r=0$. So we can parametrize the orbits by $r$, which lives on the closed interval $[0,1]$, and characterize the motion along orbits with the unit vector $\hat{n}=\vec{n}/r$. We must be careful with respect to the innermost orbit however. $\hat{n}$ is not defined there, which means that the differential equations which we will derive for $r$ and $\hat{n}$ will have solutions that exist for finite times (those solutions correspond to trajectories of $\rho$ that pass through the completely mixed state). 

Since $r^2=\vec{n}\cdot\vec{n}$ , $2r\frac{dr}{dt}=2\vec{n}\cdot\frac{d\vec{n}}{dt}$  and therefore $\frac{dr}{dt}=\hat{n}\cdot\frac{d\vec{n}}{dt}$. So:
\begin{equation*}\frac{dr}{dt} = \hat{n}\cdot \vec{b} + \hat{n}\cdot( \vec{u} \times \vec{n} ) +\hat{n}\cdot  (A^S-tr(A^S) I) \vec{n} \end{equation*}
\noindent The middle term vanishes, the first term is constant in $r$ and the third is linear in $r$. We can write:
\begin{equation}\frac{dr}{dt} = \hat{n}\cdot \vec{b} +  r (\hat{n}\cdot  (A^S \hat{n})-tr(A^S) )  \label{proj} \end{equation}
To find the ODE for $\hat{n}$, we use $\vec{n}=r\hat{n}$, which gives $\frac{d\hat{n}}{dt}=\frac{1}{r}(\frac{d\vec{n}}{dt} - \frac{dr}{dt}\hat{n})$. So we get:
\begin{align} \frac{d\hat{n}}{dt} = 2\vec{u}\times \hat{n} + \frac{1}{r}(\vec{b}-(\vec{b}\cdot\hat{n}) \hat{n}) + (A^S-\hat{n}\cdot ( A^S\hat{n}))\hat{n}\label{hat}\end{align}
Our goal here is to view equation (\ref{proj}) as a control ODE where $\hat{n}$ is the control. This view requires that we have full control over $\hat{n}$, and we claim that we do, in terms specified by the following lemma. 

\begin{lemma}
Let $S$ be the sphere centered at the origin with radius one, let $B$ be the associated closed ball, and let $B^*$ be the closed ball with the origin removed. Let $\hat{n}(t)$ be a piecewise differentiable function from a time interval $[0,T]$ onto $S$ such that the corresponding solution $r(t)$ of equation (\ref{proj}) is contained in the interval $(0,1]$. Then there are piecewise continuous control functions $u_x(t)$, $u_y(t)$ and $u_z(t)$ such that equation (\ref{unproj}) has the piecewise differentiable solution $\vec{n}(t)=r(t)\hat{n}(t)$ on $B^*$.
\end{lemma}

\begin{proof} First re-write equation (\ref{hat}):
\begin{align*} \vec{u}\times \hat{n} = \frac{1}{2} \left( \frac{d\hat{n}}{dt} - \frac{1}{r}(\vec{b}-(\vec{b}\cdot\hat{n}) \hat{n}) - (A^S-\hat{n}\cdot ( A^S\hat{n}))\hat{n}\right)\end{align*}
\noindent Any equation of the form $\vec{x}\times\vec{a} = \vec{b}$, where $\vec{a}\cdot\vec{b}=0$, has solution $\vec{x}=\vec{a}\times\vec{b}$. It follows that we can choose the controls to be:
\begin{align*}
\vec{u} (t) &=   \hat{n}\times \frac{1}{2} \left( \frac{d\hat{n}}{dt} - \frac{1}{r}(\vec{b}-(\vec{b}\cdot\hat{n}) \hat{n}) - (A^S-\hat{n}\cdot ( A^S\hat{n}))\hat{n}\right) \\
&= \frac{1}{2} \left( \hat{n}(t)\times \dot{\hat{n}} - \frac{1}{r(t)}\hat{n}(t)\times\vec{b} -\hat{n}(t)\times (A^S \hat{n}(t))\right)
\end{align*}
Since $\hat{n}(t)$, $\dot{\hat{n}}(t)$ and $r(t)$ are piecewise continuous, so is $\vec{u}(t)$. \end{proof}

Note that the prescription for $\vec{u}(t)$ is unbounded as $r\rightarrow 0$ because of the middle term. This is because the system cannot approach the completely mixed state from any direction: when $\vec{n}=\vec{0}$, $\frac{d\vec{n}}{dt}$ is fixed to be $\vec{b}$ regardless of the controls $\vec{u}(t)$.  

We finish this section by writing down an alternate version of (\ref{proj}) in terms of the eigenvalues of $A^S$, which allows us to specify a given system in terms of six real parameters. Let $a_1\ge a_2 \ge a_3$ be the eigenvalues of $A^S$. Let $\{b_j : j=1,2,3\}$ and $\{n_j : j=1,2,3\}$ be the components of $\vec{b}$ and $\vec{n}$ relative to the intrinsic axes of $A^S$ (whereas the subscripts $x$, $y$ and $z$ denote the components relative to the eigenvectors of the Pauli matrices). This gives:
\begin{equation}\frac{dr}{dt} = \sum^3_{j=1}b_jn_j -  r\sum^3_{j=1}a_j (1-n_j^2)   \label{paraproj} \end{equation}

\noindent The six parameters obey the following inequality, which arises from the positive semi-definiteness of $A$: 
\begin{equation}a_1b_1^2+ a_2b_2^2+ a_3b_3^2 \le 4a_1a_2a_3 \label{bineq}
\end{equation}
\noindent The positive semi-definiteness of $A$ also ensures the positive semi-definiteness of $A^S$, so we also have $a_1,a_2.a_3 \ge 0$.

%
%

\section{Controllability analysis}

For a fixed $r$, the right-hand side of equation \eqref{paraproj} can be seen as a map from $S^2$, the set of available controls, to the set of possible values of $\dot{r}$. Since this is a smooth map from a compact set to $\mathbb{R}$, the image should be a closed finite interval. To analyze the controllability of (\ref{paraproj}), we define functions $f_M(r)$ and $f_m(r)$ to be the right and left endpoints, respectively, of this interval. That is, $f_M(r)$ is the maximum possible rate at which $r$ can increase, and $f_(m)$ the minimum, for a given value of $r$. It is clear that (\ref{paraproj}) is controllable on a closed subinterval of $(0,1)$ if $f_M>0$ and $f_m<0$ everywhere on the subinterval. To steer between two points $r_i$ and $r_f$, we choose our controls so that $\dot{r}(t)=f_M(r(t))$ if $r_i<r_f$, or $\dot{r}(t)=f_m(r(t))$ if $r_i>r_f$. 

Some properties of $f_M$ and $f_m$ can be gleaned from inspection of the differential equation, which we collect into a proposition:
\begin{proposition} If $f_M(r):=\sup\{\dot{r}(r)\}$ and $f_m(r):=\inf\{\dot{r}(r)\}$,
\begin{enumerate}
\item $f_M$ and $f_m$ are non-increasing.
\item $\lim_{r\rightarrow 0+} f_M(r) = |\vec{b}|$ and $\lim_{r\rightarrow 0+} f_m(r) = -|\vec{b}|$.
\item $f_M(1)\le0$.
\item $f_m(r)\le 0$ for all $r$ and system parameters. $f_m(r)=0$ for $r>0$ only for the trivial where $a_1=0$ (which requires that all $a_j$'s and $b_j$'s are zero.

\item If $\vec{b}$ has non-zero magnitude, $f_M(r)$ has an isolated intercept $r_T\in (0,1]$.
\end{enumerate}
\end{proposition}
\begin{proof}
\begin{enumerate}
\item If a control vector $\hat{n}^*$ achieves the maximum $\dot{r}$ at $r=r^*$, then choosing that control for all $r<r^*$ can only achieve a larger or equal $\dot{r}$, since the coefficient of $r$ in the differential equation, $\sum^3_{j=1}a_j (1-n_j^2) $, must be non-negative. Similarly, if a control $\vec{n}^*$ achieves the minimum at $r=r^*$, then choosing that control for all $r>r^*$ can only achieve a smaller or equal $\dot{r}$. Furthermore, if $a_1$ and $a_2$ are positive, the coefficient of $r$ cannot be made zero, so in this case, we can strengthen ``non-increasing" to ``decreasing".
\item As $r\rightarrow 0+$, the linear term in (\ref{paraproj}) can be neglected, and we must extremize $\vec{b}\cdot\vec{n}$. The range of this is clearly $[-|\vec{b}|,|\vec{b}|]$
\item Since $r$ cannot exceed one, $\dot{r}|_{r=1}\le 0$.

\item Non-positivity follows from 1) and 2). If $a_1>0$, $\dot{r}$ can be always made negative by choosing $\vec{n}=\left<0,0,1\right>$. 
\item Non-zero $\vec{b}$ implies that at $a_1$ and $a_2$ are positive, which means that $f_M$ is strictly decreasing on $(0,1)$. This, together with 2) and 3) imply the existence of $r_T$.  
\end{enumerate}
\end{proof} 

\begin{corollary}
If $\vec{b}$ is nonzero, there is an interval $(0,r_T)$, which we call a \emph{trap}, inside of which the system is controllable. Outside of the trap, on $[r_T,1]$, the system is one-way controllable; that is, $r_i$ can be steered to $r_f$ in finite time if and only if $r_f\le r_i$. 
\end{corollary}

\begin{proof}The statements in the proposition imply that $f_m(r)<0<f_M(r)$ on $(0,r_T)$, in which case we can steer $r_i$ to $r_f\ge r_i$ by choosing the control that satisfies $\dot{r}=f_M(r_f)$ provided $r_f<r_T$. Conversely, to steer $r_i$ to $r_f \le r_i$, we can choose the control that satisfies $\dot{r}=f_m(r_f)$. On the interval $[r_T,1]$, $f_M(r)\le 0$, so $r_i$ cannot be steered to $r_f >r_i$, but can be steered to $r_f<r_i$ by choosing the control that satisfies $\dot{r} = f_m(r_i)$, which must be negative. \end{proof}

In the case that $|\vec{b}|=0$, there is no trap: $\dot{r}\le 0$ for all $r$, and in fact we can say that 
\begin{equation} -r(a_1+a_2) \le \dot{r} \le -r(a_2+a_3) \end{equation}
where we can achieve the upper and lower bounds by choosing $\vec{n}$ to be $\left<\pm 1,0,0\right>$ and $\left<0,0,\pm 1\right>$, respectively. In the case that $a_2=a_3=0$, the decay of $r$ may be halted, but otherwise $r$ will decay exponentially to zero at a rate above or equal to $a_2+a_3$. It is evident, then, that the presence of an asymmetric part in the dissipative term (represented by $\vec{b}$) significantly enhances the possibility of control. 

In order to calculate $f_M$ and $f_m$ for given $r$, we can use the method of Lagrange multipliers. In some cases, we can solve the resulting equations analytically, but in general one must find the roots of a sixth-order polynomial, so we must resort to numerics. Before considering the general case, we will look at a particular case that can be treated analytically. We consider the possibility that a two-level system can undergo one of two processes represented by the raising and lowering operators $\sigma_+$ and $\sigma_-$ at rates $\alpha_+$ and $\alpha_-$, respectively. If one constructs the Lindblad equation using this scenario, and expresses it in the basis of the Pauli matrices, one finds that $a_1=a_2=\frac{|\alpha_+-\alpha_-|}{2}$, $a_3=0$, $b_1=b_2=0$ and $b_3=\alpha_+-\alpha_-$. The fact that $\vec{b}$ has only one non-zero component simplifies the equations so that we can treat the system analytically. 

If we apply the method of Lagrange multipliers to the right-hand side of (\ref{paraproj}) and set $b_1=b_2=0$ and $a_1=a_2$, we get:
\begin{align*}
2ra_1n_1 &= 2 \lambda n_1 \nonumber \\ 2ra_1n_2 &= 2 \lambda n_2 \nonumber \\
b_3 &= 2 \lambda n_3 \nonumber \\
n_1^2+n_2^2+n_3^2&=1 \end{align*}
where $\lambda$ is the Lagrange multiplier. This has the following solutions:
\begin{align}
&\hat{n}= \left<0,0,\pm 1\right>  \label{sol1} \\
&\hat{n}=\left<n_1, n_2 , \frac{b_3}{2a_1r}\right>    \label{sol2} \end{align}
where $n_1$ and $n_2$ in \eqref{sol2} can be any pair that satsifies the normalization condition. Solutions \eqref{sol2} do not exist for all $r$, since the magnitude of $n_3$ must not exceed one. They exist only on $[\frac{|b_3|}{2a_1},1]$. To determine which solutions correspond to $f_M$ and $f_m$, we substitute back into \eqref{paraproj}. Solutions \eqref{sol1} give
\begin{equation} \dot{r}= \pm |b_3| - 2a_1r \label{s1} \end{equation}
and solutions \eqref{sol2} give
\begin{equation} \dot{r} = \frac{|b_3|^2}{4a_1r} - ra_1 \label{s2} \end{equation}
We can easily conclude that $f_m(r)= - |b_3| - 2a_1r$. Furthermore, the right-hand side of \eqref{s2} is greater than or equal to those of \eqref{s1}, but since it has a limited interval of definition, we have:

\begin{equation}
f_M(r)=\left\{ \begin{array}{ll}
|b_3| - 2a_1r, &            r\in(0,\frac{|b_3|}{2a_1})                       \\
\frac{|b_3|^2}{4a_1r} - ra_1,  &  r\in (\frac{|b_3|}{2a_1},1) 
\end{array} \right.
\end{equation}

It happens that $r_T$ in this case coincides with the point at which $f_M$ switches between \eqref{sol1} and \eqref{sol2}, \emph{i.e.} $r_T=\frac{|b_3|}{2a_1}$. This is not a general phenomenon, however: if $a_3>0$, the switching point and the trap radius would not coincide. Fig. \ref{analyt} depicts these solutions for $a_1=a_2=10$ and $b_3=12$. 

   \begin{figure}[thpb]
      \centering
      \includegraphics[scale=0.7]{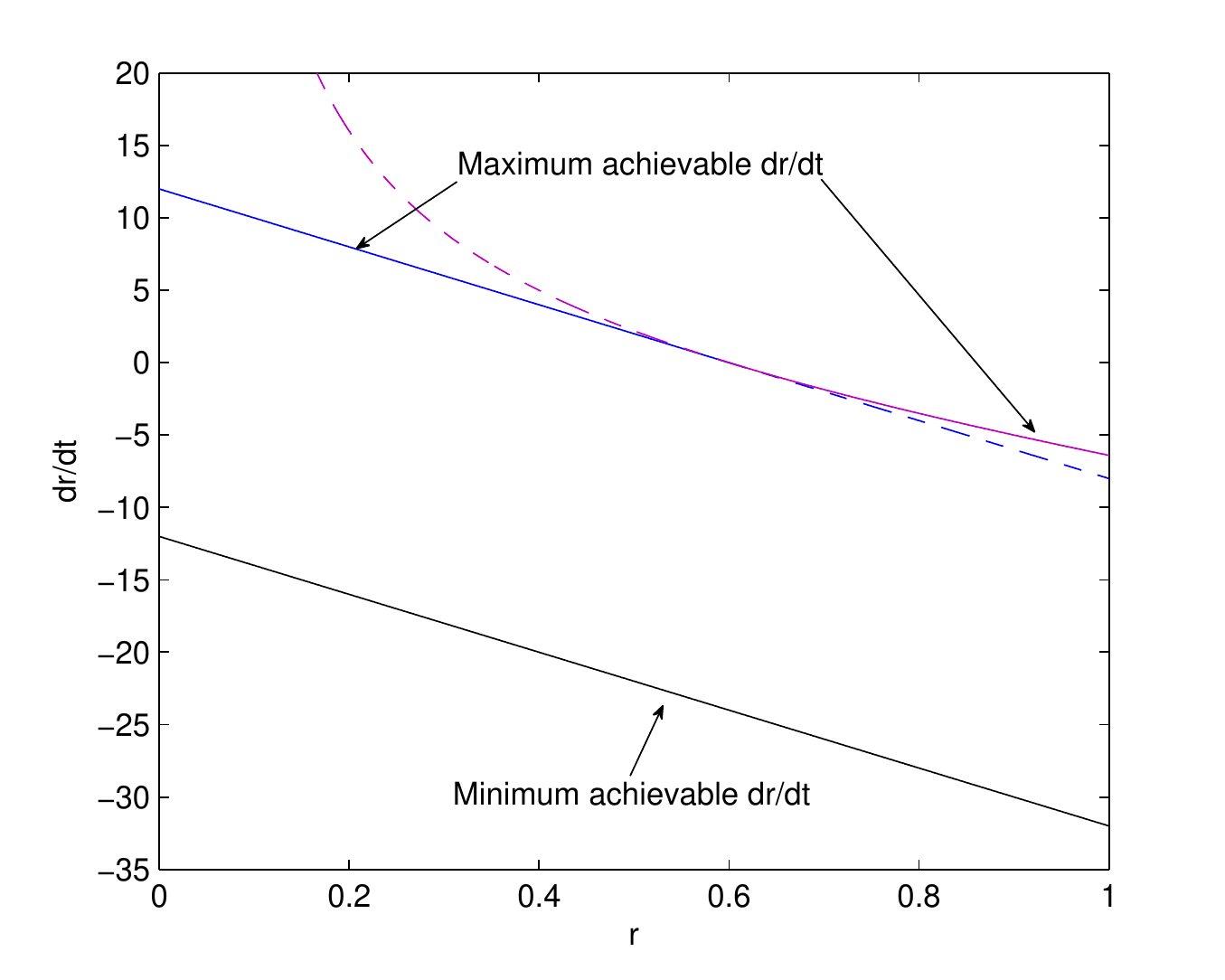}
      \caption{Maximum and minimum achievable $dr/dt$ vs. $r$ for a case that can be solved analytically. System parameters: $a_1=a_2=10$, $a_3=0$, $b_1=b_2=0$, $b_3=12$. Solid lines represent $f_M$ and $f_m$. Blue and purple indicate solutions \eqref{sol1} and \eqref{sol2}, respectively. Dotted lines indicate where these solutions do not coincide with $f_M$.}
      \label{analyt}
   \end{figure}

   \begin{figure}[thpb]
      \centering
      \includegraphics[scale=0.7]{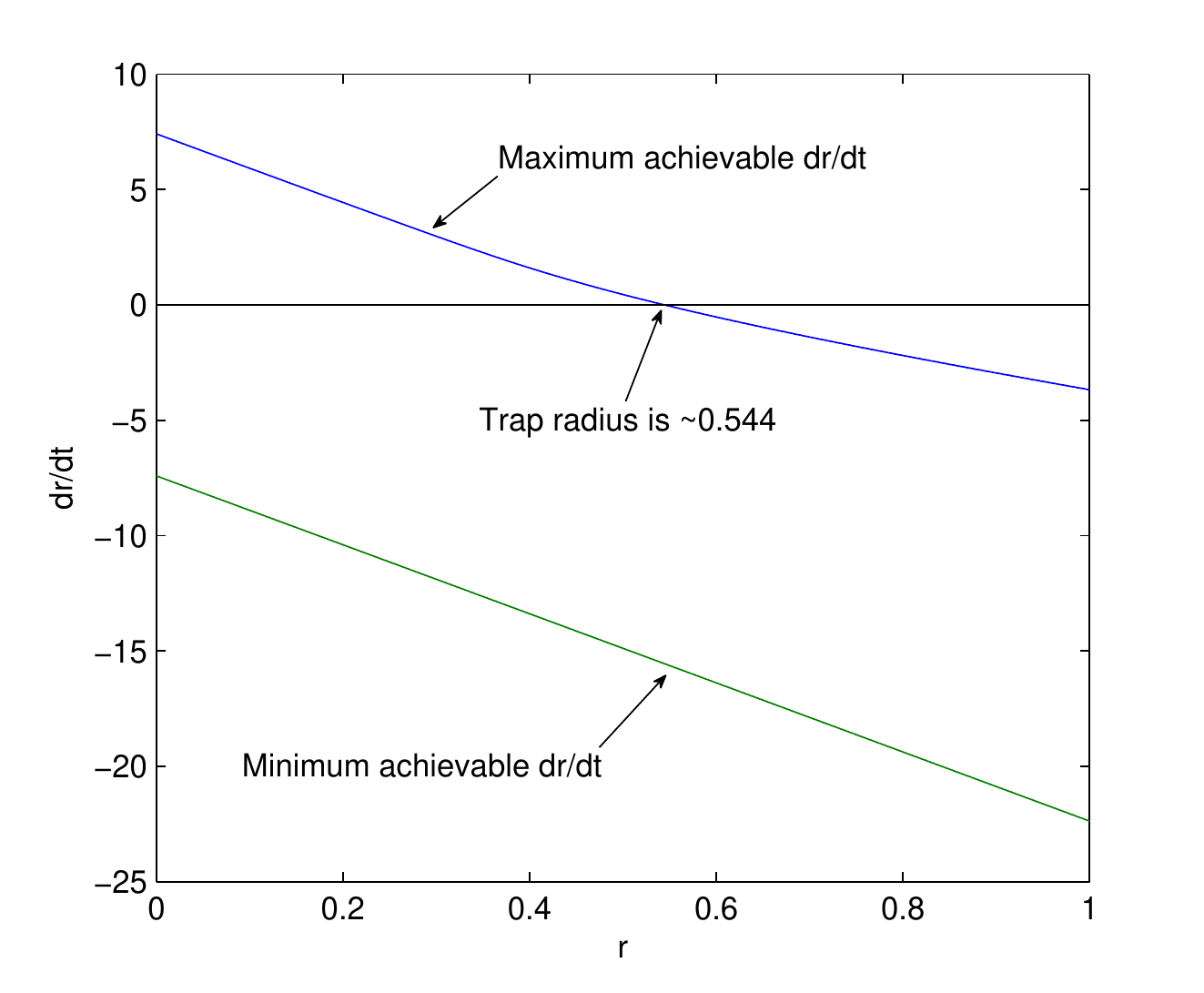}
      \caption{Maximum and minimum achievable $dr/dt$ vs. $r$ for a case that must be solved numerically. The trap radius is where the maximum achievable $dr/dt$ passes from positive to negative. System parameters: $a_1=10$, $a_2=5$, $a_3=0.3$, $b_1=0.15\sqrt{0.6}$, $b_2=0.9$, $b_3=3\sqrt{6}$}
      \label{numer}
   \end{figure}

More generally, one can perform this analytical treatment in the following cases: (1) if $\vec{b}$ has one non-zero component, (2) if $\vec{b}$ has two non-zero components, and the corresponding $a_j$'s are equal, and (3) $\vec{b}$ has three non-zero components, and $a_1=a_2=a_3$. If the system does not fall into any of those three categories, Lagrange multipliers lead to either a fourth-degree polynomial in $\lambda$ (technically solvable, but inordinately messy) or a sixth-degree polynomial (generally not solvable). The fourth-degree polynomial arises in the cases (1) $\vec{b}$ has two non-zero components but corresponding $a_j$'s are not equal and (2) $\vec{b}$ has three non-zero components and $a_1=a_2>a_3$ or $a_1>a_2=a_3$. The sixth-degree polynomial arises if $\vec{b}$ has three non-zero components and $a_1>a_2>a_3$. 

In those cases, we can find the real roots of the polynomial numerically. Then we can compute the corresponding values of $\dot{r}$, choose the maximum and minimum values, and assign the values to $f_M$ and $f_m$. In fig. \ref{numer}, $f_M$ and $f_m$ are shown for a particular system that required solving a sixth-order polynomial. We have computed the curves for 10,000 points apiece. $r_T$ can be found by numerically interpolating $f_M$. For the case depicted in fig. \ref{numer}, $r_T$ was computed to be $0.544387876644064$ (to machine precision).

\section{Purifiable systems}

An important goal in quantum control is purification: the process of steering a mixed state to a pure state, which can be  characterized by a purity $\sqrt{tr(\rho^2)}=1$. Alternatively, a system is pure if the leading eigenvalue is one, with remaining eigenvalues being zero. In terms of the above analysis, we say a system is \emph{purifiable} if and only if the trap radius $r_T=1$. In other words, the function $f_M(r)$ has an isolated intercept at $r=1$. This section is devoted to proving a theorem that characterizes the possible purifiable systems. First, we will use the following lemma:

 \begin{lemma}
$\dot{r}=0$ at $r=1$ if and only the system is in a state that is an eigenvector of all contributing Lindblad operators.
\end{lemma}

\begin{proof}
Because the Bloch radius can be written $r=\lambda_+-\lambda_-$, where $\lambda_+ \ge \lambda_2$ are the eigenvalues of $\rho$, we can write $r= \langle\psi_+|\rho|\psi_+\rangle - \langle\psi_-|\rho|\psi_-\rangle$. Differentiating this, we get an alternative expression for $\dot{r}$:
\begin{align*}
\frac{dr}{dt} &=  \langle\dot{\psi_+}|\rho|\psi_+\rangle - \langle\dot{\psi_-}|\rho|\psi_-\rangle
			+ \langle\psi_+|\dot{\rho}|\psi_+\rangle - \langle\psi_-|\dot{\rho}|\psi_-\rangle 
			+ \langle\psi_+|\rho|\dot{\psi_+}\rangle - \langle\psi_-|\rho|\dot{\psi_-}\rangle       \\
 &=  \lambda_+ ( \langle\dot{\psi_+}|\psi_+\rangle + \langle\psi_+|\dot{\psi_+}\rangle) 
 - \lambda_-  (\langle\dot{\psi_-}|\psi_-\rangle  + \langle\psi_-|\dot{\psi_-}\rangle )
			+ \langle\psi_+|\dot{\rho}|\psi_+\rangle - \langle\psi_-|\dot{\rho}|\psi_-\rangle \\
&= \langle\psi_+|\dot{\rho}|\psi_+\rangle - \langle\psi_-|\dot{\rho}|\psi_-\rangle 
\end{align*}
\noindent where in the last step, the normalization of the vectors makes the quantities in parentheses vanish. 
Now, if the dissipation is characterized by a collection of Lindblad operators $\{ L_j\}$'s, which are not necessarily orthogonal we can use \eqref{lindblad} to specify $\dot{\rho}$:
\begin{align*}
\frac{dr}{dt} &= \langle\psi_+|[-iH,\rho]|\psi_+\rangle - \langle\psi_-|[-iH,\rho]|\psi_-\rangle \\
&+ \sum_j  \left( \langle\psi_+|L_j\rho L_j^\dagger|\psi_+\rangle -\frac{1}{2}\langle\psi_+|L_j^\dagger L_j\rho |\psi_+\rangle-\frac{1}{2}\langle\psi_+|\rho  L_j^\dagger L_j|\psi_+\rangle \right.\\
& \left. - \langle\psi_-|L_j\rho L_j^\dagger|\psi_-\rangle + \frac{1}{2}\langle\psi_-|L_j^\dagger L_j\rho |\psi_-\rangle+\frac{1}{2}\langle\psi_-|\rho  L_j^\dagger L_j|\psi_-\rangle  \right) \end{align*}

\noindent The Hamiltonian terms vanish since they are diagonal elements of a skew-symmetric matrix. We are interested in $\dot{r}$ when $r=1$, so insert $\rho=|\psi_+\rangle\langle\psi_+|$. We get:
\begin{align*}
\frac{dr}{dt} = \sum_j \left( \langle\psi_+|L_j|\psi_+\rangle\langle\psi_+| L_j^\dagger|\psi_+\rangle-\langle\psi_+| L_j^\dagger L_j|\psi_+\rangle  - \langle\psi_-|L_j|\psi_+\rangle\langle\psi_+| L_j^\dagger|\psi_-\rangle \right)
\end{align*} 
\noindent If we insert the identity operator between $L^\dagger_j$ and $L_j$ in the middle term, we get the expression:
\begin{align*}
\frac{dr}{dt} = -2\sum_j |\langle \psi_-| L_j | \psi_+\rangle|^2
\end{align*}
For $\dot{r}$ to vanish, we need $|\langle \psi_-| L_j | \psi_+\rangle|^2$ to vanish for each $L_j$. This is only possible however if $|\psi_+\rangle$ is an eigenvector of each $L_j$, since otherwise $L_j|\psi_+\rangle$ would have some component in the $|\psi_-\rangle$ direction. This proves the lemma.
\end{proof}

This leads to the following theorem:

\begin{theorem}
A two-level Lindblad system is purifiable if and only if one of the following characterizations hold:
\begin{itemize}
\item There is one Lindblad operator, and it is singular.
\item There is one Linblad operator and it is non-singular with non-orthogonal eigenvectors.
\item There is no more than one singular Lindblad operator and any number of non-singular operators. All share a common eigenvector.  
\item There are any number of non-singular Lindblad operators that share a common eigenvector.
\end{itemize}
\end{theorem} 

\begin{proof}
We are required to show two things to prove a system is purifiable: $f_M(1)=0$, and $a_2>0$. $a_2 \ge 0$. The latter ensures that $f_M$ is strictly decreasing rather than constant in $r$. When combined with the former condition, this implies that $f_M$ is positive for all $r<1$, and therefore controllable. 

It follows from the lemma that a system is purifiable only if all contributing Lindblad operators share a common eigenvector, or else $f_M(1)$ will be strictly negative. This is only a necessary condition however and not a sufficient one, since the condition implies only that $f_M(1)=0$. We also require that $a_2>0$. So consider the case $a_2=0$. This implies that $a_3$ and $\vec{b}$ are also zero (due to \eqref{bineq}), so that $A$ has only one non-zero entry in its natural basis. We claim that $A$ in this form corresponds to a non-singular operator with orthogonal eigenvectors. It is a rank-one real positive matrix, and therefore can be written $A=\sum_{ij=x,y,z}m_im_j$ for some real 3-vector $\vec{m}$. When one diagonalizes the Lindblad equation however, this results in a single Lindblad operator $L=\sum_{j=x,y,z} m_j \sigma_j$. This operator is Hermitian and traceless, however, so neglecting the trivial zero operator, it is non-singular with orthogonal eigenvectors. 

In other words, as long as the system obeys the terms of the lemma, and does not consist of a single Hermitian operator, the system is purifiable. The first two cases in the theorem cover the remaining possible single-operator cases. The remaining two cases can be seen by noting that two singular operators cannot share eigenvectors, since they have only one (we consider two operators that are multiples of each other to be essentially one process). The third case covers the possibility of one singular operator: it has only one eigenvector, and that eigenvector must be shared with the other non-singular operator. The fourth case in the theorem covers the possibility of no singular operators but more than one non-singular operator. Note that the non-singular operators in the third and fourth cases need not have non-orthogonal eigenvectors. 
\end{proof}

\section{Conclusions}

We have shown that the inter-orbit dynamics of a controlled quantum system can be isolated from the intra-orbit dynamics by projecting onto the set of spectra of the density matrix. If one makes certain assumptions about the controllability of the system along the orbits, the position of the system along the orbit can be viewed as a new control variable, since the intra-orbit dynamics can be made arbitrarily faster than the inter-orbit dynamics. In two dimensions, we have derived a 
 equation describing this inter-orbit dynamics, where the new control is the normalized Bloch vector, and the most general Lindblad system can be described by six real parameters: three describing the symmetric part of the dissipation, and three describing the anti-symmetric part.

We have analyzed the controllability of a general two-dimensional system under Lindblad dissipation, particularly for dissipation with non-zero anti-symmetric part. For systems of this type, there exists a trap, or a subinterval of the state space where each state is reachable from any other, but from which states may not escape. The size of this trap can be calculated analytically for certain simple cases, but in general must be calculated numerically. We have shown how this can be done using the method of Lagrange multipliers, and shown results for a particular generic system. 

Furthermore, we have applied this formalism to categorize the set of purifiable systems. A necessary condition for purifiability is that all Lindblad operators share a common eigenvector. To strengthen this to a sufficient condition, one must eliminate the case of a single Hermitian Lindblad operator.

The immediate direction of future work is to apply this formalism to three and higher dimensional systems. It is well-known that the structure of density matrices is richer and less well-understood than the case for two dimensions \cite{schirmer}. For one, the set of pure states no longer constitutes the boundary of the set, but a (measure-zero) subset of the boundary. Currently we are studying how our formalism translates to higher dimensions and what obstructions are posed by the richer geometry. Furthermore, we would like to know whether the set of purifiable higher-dimensional systems can be categorized as it has been done in this paper. 


\end{document}